\documentclass[reqno, centertags,cmex10]{amsart}
\usepackage{amsfonts,amssymb,curves,latexsym,hyperref}
\usepackage[dvipsnames,svgnames,x11names,hyperref]{xcolor}
\hypersetup{colorlinks,breaklinks,
             urlcolor=blue,
             linkcolor=blue}
\newtheorem{theorem}{Theorem}
\newtheorem{proposition}[theorem]{Proposition}
\newtheorem{lemma}[theorem]{Lemma}
\newtheorem{example}[theorem]{Example}
\newtheorem{remark}[theorem]{Remark}
\newtheorem{corollary}[theorem]{Corollary}
\newtheorem{definition}[theorem]{Definition}

\newcommand{\rk}{\mathop{\rm rank}}
\def\Fq{{\mathbb F}_q}
\def\Fqm{{\mathbb F}_{q^m}}
\def\Fqmn{{\mathbb F}_{q^m}^n}
\def\rank{\mathop{\rm rank}}
\def\M{{\mathbb M}}
\def\Mat{{\mathbb M}_{m\times n}(\Fq)}
\def\F{\mathsf{F}}
\def\rhoF{\rho_{\rule{0pt}{2ex}\mathsf{F} }}
\def\rhoG{\rho_{\rule{0pt}{2ex}\mathsf{G} }}
\def\rhoFp{\rho_{\rule{0pt}{2ex}\mathsf{F}' }}
\def\rhoFperp{\rho_{\rule{0pt}{1.8ex}\mathsf{F}^\perp }}

\def\nuF{\nu_{\rule{0pt}{1.5ex}\mathsf{F} }}
\def\nusF{\nu^*_{\rule{0pt}{1.5ex}\mathsf{F} }}
\def\PF{P(\F)} 

\begin{document}

\title[A Polymatroid Approach to Rank Metric Codes]{A polymatroid approach to generalized  weights of Rank Metric Codes}

\author[Ghorpade]{Sudhir R. Ghorpade}
\address{Department of Mathematics, 
Indian Institute of Technology Bombay  \newline \indent
Powai, Mumbai 400\,076, India}
\email{srg@math.iitb.ac.in} 
\thanks{Sudhir Ghorpade is partially supported by DST-RCN grant INT/NOR/RCN/ICT/P-03/2018 from the Dept. of Science \& Technology, Govt. of India, MATRICS grant MTR/2018/000369 from the Science and Engg. Research Board, and  IRCC award grant 12IRAWD009 from IIT~Bombay.}

\author[Johnsen]{Trygve Johnsen}
\address{Department of Mathematics and Statistics, 
 UiT-The Arctic University of Norway  \newline \indent 
N-9037 Troms{\o}, Norway}
\email{trygve.johnsen@uit.no}
\thanks{Trygve Johnsen is partially supported by grant 280731 from the Research Council of Norway.}

\date{\today}

\begin{abstract}
We consider the notion of a  $(q,m)$-polymatroid, due to Shiromoto, and the more general notion of  $(q,m)$-demi-polymatroid, and show how generalized  weights can be defined for them. Further, we establish a duality for these weights analogous to Wei duality for generalized Hamming weights of linear codes. The corresponding results of Ravagnani for Delsarte rank metric codes,  and Mart\'{\i}nez-Pe\~{n}as and Matsumoto for relative generalized rank weights are derived as a consequence. 
\end{abstract}

\maketitle


\section{Introduction}
Linear (block) codes are objects of basic importance in the theory of error correcting codes. A $q$-ary linear code of length $n$ and dimension $k$, or in short, a $[n,k]_q$-code $C$ is simply a $k$-dimensional subspace of $\Fq^n$, where $\Fq$ denotes the finite field with $q$ elements. A basic notion here is that of Hamming distance on the space $\Fq^n$, which for two vectors $x, y \in \Fq^n$ is simply the number of nonzero coordinates in $x-y$. Rank metric codes are an important variant of linear codes, and they have gained prominence in the past few decades, partly due to myriad applications in network coding and cryptography, as also due to their intrinsic interest. Perhaps a more widely studied notion of rank metric codes is the one that goes back to Gabidulin's work \cite{Gab85} in 1985. 
A \emph{Gabidulin rank metric code}, or simply, a \emph{Gabidulin code}, of length $n$ and dimension $k$ may be defined as a $k$-dimensional 
subspace of the $n$-dimensional vector space $\Fqmn$ over the extension field  $\Fqm$ of $\Fq$. By fixing a $\Fq$-basis of $\Fqm$, we can associate to any vector in 
$\Fqmn$ an $m\times n$ matrix with entries in $\Fq$, and the {rank distance} between any $\mathbf{x}, \mathbf{y} \in \Fqmn$ is defined as the rank of the difference of the 
matrices corresponding to $\mathbf{x}$ and  $\mathbf{y}$. The notion of a Delsarte rank metric code is in fact, older (it goes back to the work \cite{Del} of Delsarte in 1978) and more general. Indeed, a 
\emph{Delsarte rank metric code}, or simply, a \emph{Delsarte code} of dimension $K$ is a $K$-dimensional subspace of the $\Fq$-linear space of all $m\times n$ matrices with entries in $\Fq$. As before, the rank distance between two $m\times n$ matrices is the rank of their difference. It is clear that a Gabidulin code of dimension $k$ is a Delsarte code of dimension $mk$. But a Delsarte code need not be a Gabidulin code, even if its dimension is divisible by $m$. 

Generalized Hamming weights (GHW), also known as higher weights, of a linear code $C$ are a natural and useful  generalization of the basic notion of minimum distance of $C$. These were studied by Wei \cite{wei91} who showed that the GHW $d_1, \dots , d_k$ of a $[n,k]_q$-code $C$ satisfy nice properties such as monotonicity ($d_1< \cdots < d_k$) and more importantly, duality whereby the GHW of $C$ and its dual $C^\perp$ determine each other. It was not immediately clear how an analogue of GHW for rank metric codes could be defined. 
But then three different definitions for the \emph{generalized rank weights} (GRW) of a Gabidulin rank metric code were proposed 
by three sets of authors working in different parts of the globe, 
viz., Oggier and Sboui \cite{OS}, Kurihara, Matsumoto and  Uyematsu \cite{KMU}, and Jurrius and Pellikaan \cite{JP2017}.  Thankfully, all three seemingly disparate definitions turn out to be equivalent (cf. \cite{JP2017}). Moreover, an analogue of Wei duality holds for the GRW; see, e.g., 
Ducoat \cite{Duc}. For the more general class of Delsarte rank metric codes,  Ravagnani \cite{ravagnani16} proposed an analogous definition of \emph{generalized weights} (GW) and showed that in the special case of Gabidulin codes, the $km$ GW of the corresponding Delsarte code are 
the same as the $k$ GRW of the Gabidulin code (in accordance with the previous definitions),  each repeated $m$ times. 
Further, Ravagnani \cite{ravagnani16} established a duality for his 
GW of Delsarte rank metric codes. The notion of dual Delsarte codes is facilitated by the \emph{trace product}, which associates to a pair $(A,B)$ of $m\times n$ matrices with entries in $\Fq$ the element $\mathrm{Trace}(AB^t)$ of $\Fq$. It is shown by Ravagnani \cite{ravagnani15} that for suitable choices of $\Fq$-bases of $\Fqm$, the notions of the (standard) dual of a  Gabidulin code and of the (trace) dual of the corresponding Delsarte code are compatible. 

In the classical case of linear codes, Britz et al \cite{BJMS} showed that Wei duality for generalized Hamming weights of linear codes is, in fact, a special case of Wei duality for matroids and also established Wei-type duality theorems for demi-matroids. It is natural, therefore, to ask if the notion of generalized (rank) weights for (Gabidulin or Delsarte)
rank metric codes can be studied in the more general context of something like matroids, and if an analogue of Wei duality can be proved in this set-up. This is the question that we address in this paper. The notion that turns out to be relevant for us is that 
of a $(q,m)$-polymatroid, which has recently been introduced by Shiromoto \cite{S}. 
(See also \cite{GJLR, G} for an essentially equivalent notion of a $q$-polymatroid.) 
Thus, we define 
generalized weights for $(q,m)$-polymatroids, and establish a Wei-type duality for them. As a corollary, we readily obtain 
the results of Ravagnani \cite{ravagnani16} for his GW of Delsarte codes and their duals, 
provided $m >n$.
The cases $m=n$ and $m<n$ can also be covered, 
and these are addressed in Remark \ref{imprem} and Proposition \ref{demisquare}, respectively.
To study the case $m=n$, and also for other purposes, we consider the more general class of $(q,m)$-demi-polymatroids and establish a duality result there. As another important application, we show how these general combinatorial objects 
can be applied to 
flags, or chains, of Delsarte rank metric codes. 
In particular, by considering pairs, i.e., flags of length $2$ of Delsarte codes, we recover several results of 
Mart\'{\i}nez-Pe\~{n}as and Matsumoto \cite{MM} on the so called relative generalized rank weights of Delsarte 
codes. We remark that $q$-analogues of matroids, called $q$-matroids and $q$-polymatroids, have been considered by Jurrius and Pellikaan \cite{JP} and by Gorla, Jurrius, Lopez, and Ravagnani \cite{GJLR}, respectively. 
However, as far as we can see, Wei-type duality for the generalised weights of these objects is not shown in these papers. 

This paper is organized as follows. In Section \ref{sec:defnnot} below, we review the definition of a $(q,m)$-polymatroid and outline some basic notions and results. Generalized weights of  a $(q,m)$-polymatroid are defined and Wei-type duality for them is established in Section \ref{sec3}. These results are then applied to Delsarte rank metric codes in Section \ref{sec4}. 
{In Section \ref{sec5} we introduce $(q,m)$-demi-polymatroids, show Wei duality for these objects, and apply it to Delsarte rank metric codes consisting of square matrices. Flags of Delsarte rank metric codes, and their dualøity theory, are discussed in Section \ref{sec5plus}.  Several examples and applications are also included here. 

\section{Preliminaries about $(q,m)$-Polymatroids}
\label{sec:defnnot}

Throughout this paper $\mathbb{N}_0$ denotes the set of all nonnegative integers, $m,n$ denote positive integers, $q$ a prime power, and $\Fq$ the finite field with $q$ elements. 
We let $E$ be the vector space $\mathbb{F}_q^n$ over $\Fq$ and let 
$$
\Sigma(E) = \text{the set of all $\Fq$-linear subspaces of $E$}.
$$
For $X\in \Sigma(E)$, we deonte by $X^{\perp}$ the dual of $X$ (with respect to the standard ``dot product"), i.e., 
$X^{\perp} = \{\mathbf{x} \in E: \mathbf{x}  \cdot \mathbf{y}  = 0 \text{ for all } \mathbf{y} \in X\}$. 
It is elementary and well-known that $X^{\perp}\in \Sigma(E)$ with $\dim X^{\perp} = n - \dim X$ and $(X^{\perp})^{\perp} = X$, although $X\cap X^{\perp}$ need not be equal to $\{\mathbf{0}\}$, but of course $E^\perp = \{\mathbf{0}\}$. 

The following key defnition 
is due to Shiromoto \cite[Definition 2]{S}. 

\begin{definition}
{\rm 
A $(q,m)$-{\em polymatroid} is an ordered pair $P= (E,\rho)$ consisting of
the vector space $E=\mathbb{F}_q^n$ 
 and a function $\rho:\Sigma(E)\to\mathbb{N}_0$
satisfying the following three conditions for all $X,Y \in \Sigma(E)$:  
\begin{list}{}{\leftmargin=1em\topsep=1.5mm\itemsep=1mm}
\item[{\rm (R1)}]  $0\leq \rho(X) \leq m\dim X$;
\item[{\rm (R2)}] If $X \subseteq Y$, then $\rho(X) \le \rho(Y)$;
\item[{\rm (R3)}] $\rho(X+Y)+\rho(X \cap Y) \le \rho(X)+\rho(Y)$.
\end{list}
The nonnegative integer $\rho(E)$ is called the \emph{rank} of $P$ and is denoted by $\rank P$. 
The function $\rho$ may be called the \emph{rank function} of $P$.
}
\end{definition}


\begin{remark}
{\rm
As Shiromoto \cite{S} remarks, a $(q,m)$-polymatroid is a $q$-analogue of $k$-polymatroids, and a $(q,1)$-matroid is a $q$-analogue of matroids. An alternative, but somewhat different approach, to $(q,m)$-polymatroids is provided by  Gorla, Jurrius, Lopez, and Ravagnani \cite[Definition 4.1.]{GJLR}. 
}
\end{remark}

The following basic fact is proved in \cite{S}. 

\begin{proposition} \label{dualdef}  \cite[Proposition 5]{S}
Let $P= (E,\rho)$ be a $(q,m)$-{polymatroid}. Define $\rho^*:\Sigma(E) \to \mathbb{N}_0$ by 
$$
\rho^*(X)=\rho(X^{\perp})+m\dim X-\rho(E)
\quad \text{for } X\in \Sigma (E).
$$
Then $(E,\rho^*)$ is also a $(q,m)$-polymatroid.
\end{proposition}

If $P= (E,\rho)$ and $\rho^*$ are as in Proposition~\ref{dualdef}, 
then the $(q,m)$-polymatroid $(E,\rho^*)$ 
is denoted  by $P^*$ and called   the 
\emph{dual} 
of $P$. Note that $\rho(\{\mathbf{0}\}) = 0$ by (R1) and so 
$$
\rk P^* = \rho^*(E)=\rho(\{\mathbf{0}\})+m\dim E-\rho(E)=mn-\rk P \quad \text{and}\quad (P^*)^*=P.
$$
\begin{definition}
{\rm 
Let $P= (E,\rho)$ be a $(q,m)$-{polymatroid}. 
The \emph{nullity function} of $P$ is the map $\nu: \Sigma(E)\longmapsto\mathbb{N}_0$  defined by 
$$
{\nu}(X)=m\dim X - \rho(X) \quad \text{for } X\in \Sigma (E).
$$
The \emph{conullity function} of $P$ is the map 
$\nu^*: \Sigma(E)\longmapsto\mathbb{N}_0$  defined by 
$$
{\nu}^*(X)=m\dim X -\rho^*(X)= \rho(E)-\rho(X^{\perp}) \quad \text{for } X\in \Sigma (E). 
$$
}
\end{definition}


By way of giving an example of a $(q,m)$-polymatroid, we describe below an important class of $(q,m)$-polymatroids. 

\begin{example}
{\rm 
Let $r$ be a nonnegative integer $\le n$. 
The \emph{uniform $(q,m)$-polymatroid} $U(r,n)$ is defined as $(E,\rho)$, where $E=\mathbb{F}_q^n$, and $\rho(X)=m\dim X,$ for all $X\in \Sigma(E)$ with $\dim X \le r$, while $\rho(X)=mr$ for all $X\in \Sigma(E)$ with $\dim X \ge r$. It is easy to see that $U(r,n)$ is indeed a $(q,m)$-polymatroid and also that $U(r,n)^*=U(n-r,n).$
}
\end{example} 

Elementary properties of nullity and conullity functions are given 
below. The 
proof is analogous to \cite[Lemma 4]{S}, but included  for the convenience of the reader. 

\begin{proposition} \label{tech}
Let $P= (E, \rho)$ be a $(q,m)$-{polymatroid} and let $X,Y \in \Sigma (E)$ with $X\subseteq Y$.  Then: 
\begin{enumerate}
\item[{\rm (a)}]  ${\nu}(X) \le {\nu}(Y)$ and 
$\, {\nu}^*(X) \le {\nu}^*(Y)$; 
\item[{\rm (b)}] ${\nu}(Y)-{\nu}(X)  \le m(\dim Y - \dim X)$ and 
$\, {\nu}^*(Y)-{\nu}^*(X)  \le m(\dim Y - \dim X)$.
\end{enumerate}
\end{proposition}

\begin{proof}
(a) By extending a basis of $X$ to $Y$, we can find $Z\in \Sigma(E)$ such that 
$$
X+Z = Y, \quad X\cap Z = \{\mathbf{0}\}, \quad \text{and} \quad \dim Z = \dim Y - \dim X.
$$
Thus, using (R3), we obtain 
$$
\rho(X)+\rho(Z)\ge \rho(X+Z)+\rho(X \cap Z)=\rho(Y)+\rho(\{\mathbf{0}\})=\rho(Y).
$$
On the other hand, by (R1), $\rho(Z)\le m\dim Z = m(\dim Y - \dim X)$. It follows that 
$$
\rho(Y)-\rho(X) \le m(\dim Y - \dim X).
$$
This proves that ${\nu}(X) \le {\nu}(Y)$. Replacing $P$ by $P^*$, we obtain ${\nu}^*(X) \le {\nu}^*(Y)$. 
%

(b) The desired upper bound for ${\nu}(Y)-{\nu}(X)$ follows by 
noting that by (R2), 
$$
{\nu}(Y)-{\nu}(X) = m\left( \dim Y - \dim X \right) + \rho(X) - \rho(Y) \le m\left( \dim Y - \dim X \right).
$$
As in (a), the inequality for $\nu^*$ follows from using $P^*$ in place of $P$. 
%
\end{proof}

\begin{remark}
\label{exa:not_polymatroid}
{\rm 
Proposition~\ref{tech} shows that if $P= (E, \rho)$ is a $(q,m)$-{polymatroid}, then the conullity function $\nu^*$ of $P$ is a monotonically increasing function on $\Sigma(E)$ (ordered by inclusion of subspaces of $E$) and it takes values ranging from ${\nu}^*(\{\mathbf{0}\})=0$ to ${\nu}^*(E)=
\rho(E)-\rho(E^{\perp})=\rho(E)-\rho(\{\mathbf{0}\})=\rho(E) = \rk P.$ However, unlike in the case of usual matroids, there is no ``discrete intermediate value theorem" saying that every integer value between $0$ and $\rho(E)$ is
attained as the conullity of some subspace of $E$. Moreover, although Proposition~\ref{tech} shows that the pairs $(E, \nu)$ and $(E, \nu^*)$ satisfy the axioms (R1) and (R2) in the definition of a $(q,m)$-{polymatroid}, 
neither of these are, in general, $(q,m)$-{polymatroids}. To see these two assertions, it suffices to consider the uniform $(q,m)$-polymatroid $U(1,2)$.   Indeed, in this case $U(1,2) = (E,  \rho)$, where $E= \Fq^2$, and it is easily seen that for any subspace $X$ of $E$, 
$$
\rho(X) = \begin{cases} 0 & \text{if } X= \{\mathbf{0}\}, \\ m &  \text{if } X \ne \{\mathbf{0}\} \end{cases} \quad \text{and} \quad 
\nu(X) = \nu^*(X) = \begin{cases} 0 & \text{if } X \ne E, \\ m &  \text{if } X = E. \end{cases} 
$$
Thus, a ``discrete intermediate value theorem" does not hold for $\nu$ as well as for $\nu^*$ if $m>1$. Furthermore, if $X,Y$ are distinct $1$-dimensional subspaces of $E= \Fq^2$, then $X+Y = E$ and $X+Y=\{\mathbf{0}\} $, and hence
$$
\nu(X+Y) + \nu(X\cap Y) = m  \not\le 0 = \nu(X) + \nu(Y).
$$
It follows that neither $(E, \nu)$ nor $(E, \nu^*)$ is a $(q,m)$-polymatroid.  
}  
\end{remark}

%

\section{Wei duality of $(q,m)$-polymatroids}
\label{sec3}

The following definition for the generalized weights of a $(q,m)$-polymatroid appears to be natural. 

\begin{definition}
{\rm
Let $P=(E,\rho)$ be a $(q,m)$-polymatroid and let $K= \rk P$. 
For $r=1,\dots,K$, the 
\emph{$r$th generalized  weight} of $P$ is 
defined by 
$$
d_r(P)=\min \{\dim X : X\in \Sigma (E) \text{ with }  {\nu}^*(X) \ge r \}.
$$
}
\end{definition}

Here are some simple properties of generalized weights of $(q,m)$-polymatroids. 

\begin{proposition} \label{WeakMono}
Let $P=(E,\rho)$ be a $(q,m)$-polymatroid and let $K= \rk P$. Then 
$$
1\le d_r(P) \le d_{r+1}(P) \le n \quad \text{for } 1\le r < K. 
$$
\end{proposition}
\begin{proof}
Since ${\nu}^*(\{\mathbf{0}\})=0$, it is clear that $1\le d_r(P) \le n$ for $1\le r\le K$. Next, if $1\le r < K$ and if $d_{r+1}(P) = \dim Y$ for some $Y\in \Sigma(E)$ with $\nu^*(Y)\ge r+1$, then $\nu^*(Y) \ge r$, and so by definition, $d_r(P) \le \dim Y = d_{r+1}(P)$. 
\end{proof}

Unlike the generalized Hamming weights of linear codes, strict monotonicity may not hold for generalized weights of $(q,m)$-polymatroids, i.e., we may not have $d_r(P) < d_{r+1}(P)$ for $1\le r < K$. For example, if $K>n$, then  Proposition~\ref{WeakMono} implies that $d_r(P) = d_{r+1}(P)$ for some $r< K$. 
However, 
we will show that $d_r(P) < d_s(P)$ for $1\le r < s\le K$, provided $s-r\ge m$. First, we need some preliminary results. 

\begin{lemma}\label{hx}
Let $P=(E,\rho)$ be a $(q,m)$-polymatroid and let $K= \rk P$. 
Define 
\begin{eqnarray*}
h(x) &=& \max\{\nu(X) : X\in \Sigma (E) \text{ with } \dim X = x\},  \text{ and} \\ 
h^*(x) &=& \max\{\nu^*(X) : X\in \Sigma (E) \text{ with } \dim X = x\} \quad \text{for } x=0,1, \dots , n.
\end{eqnarray*}
Now fix a positive integer $x \le n$. 
Then $h^*(x-1) \le  h^*(x)$ and 
for $1\le r\le K$,
\begin{equation}\label{drP}
x = d_r(P) \Longleftrightarrow h^*(x-1) < r \le h^*(x) 
\end{equation}
In particular, $x$ is a generalized weight of $P$ if and only if $h^*(x-1) <  h^*(x)$. Also,  $h(x-1) \le  h(x)$ and if $P^*$ is the dual of $P$, then  for $1\le s \le \rk P^*=mn-K$, 
\begin{equation}\label{dsPstar}
x = d_s(P^*) \Longleftrightarrow h(x-1) < s \le h(x)
\end{equation}
In particular, $x$ is a generalized weight of $P^*$ if and only if $h(x-1) <  h(x)$.
\end{lemma}

\begin{proof}
Let $x\in \mathbb{N}_0$ 
with $1\le x \le n$. If $X\in \Sigma(E)$ is such that $\dim X = x-1$ and $h^*(x-1) = \nu^*(X)$, then by taking $Y\in \Sigma (E)$ with $\dim Y = x$ and $X \subset Y$, we see from Proposition~\ref{tech} (a) that 
$\nu^*(X) \le \nu^*(Y) \le h^*(x)$. Thus, $h^*(x-1) \le  h^*(x)$. Similarly, $h(x-1) \le  h(x)$. Now let 
$r\in \mathbb{N}_0$ 
with $1\le r\le K$.

First, suppose $x=d_r(P)$. 
Then $x = \dim Y$ for some $Y\in \Sigma (E)$ with $\nu^*(Y) \ge r$. This implies that $h^*(x) \ge r$. 
Moreover, since $x=d_r(P)$, we see that $\nu^*(X) < r$
for every $X\in \Sigma (E)$ with $\dim X = x-1$. This implies that $h^*(x-1) < r$. 

Conversely, suppose $h^*(x-1) < r \le h^*(x)$. Choose $Y\in \Sigma (E)$ with $\dim Y = x$ such that $h^*(x) = \nu^*(Y)$. Then $\nu^*(Y) = h^*(x) \ge r$ and so $d_r(P)\le x$. 
Suppose, if possible, $d_r(P)\le x-1$. Then  there is $Z\in \Sigma(E)$ with $\dim Z = d_r(P) \le x-1$ and $\nu^*(Z) \ge r$. Enlarge $Z$ to a subspace $X$ of $E$ such that $\dim X = x-1$. In view of Proposition~\ref{tech} (a), we  obtain  $h^*(x-1) \ge \nu^*(X) \ge \nu^*(Z) \ge r $, which contradicts the assumption  $h^*(x-1) < r $. This shows that $x=d_r(P)$. Thus \eqref{drP} is proved. 

The equivalence \eqref{dsPstar} follows by applying \eqref{drP} to $P^*$ in place of $P$. 
\end{proof}

Here is a nice relation between the functions $h$ and $h^*$ defined in Lemma~\ref{hx}. 

\begin{lemma}\label{hhstar}
Let $P=(E,\rho)$ be a $(q,m)$-polymatroid and let $K= \rk P$. 
Then 
\begin{equation}\label{basich}
h^*(x) = h(n-x) - m (n-x) + K \quad \text{for } x=0,1, \dots , n.
\end{equation}
Consequently, 
\begin{equation}\label{diffh}
h(n+1-x) - h(n-x) = m - \left( h^*(x) - h^*(x-1) \right) \quad \text{for } x=1, \dots , n.
\end{equation}
In particular, $0 \le h^*(x) - h^*(x-1) \le m$ for $x=1, \dots , n$. 
\end{lemma}

\begin{proof} 
Given any $X\in \Sigma(E)$, note that $\nu (X^\perp) = m \dim X^\perp - \rho(X^\perp)$, and hence
$\nu (X^\perp) + m \dim X = mn - \rho(X^\perp) = mn -  \left( \rho(E) - \nu^*(X) \right)$. 
It follows that 
$$
\nu^*(X) = \nu (X^\perp) - m (n - \dim X) + K.
$$
Taking maximum as $X$ varies over elements of $\Sigma(E)$ with $\dim X = x$, we obtain \eqref{basich}. 
Now \eqref{basich} implies that $h^*(x) - h^*(x-1) = h(n-x) - h(n+1-x) + m$ for $x=1, \dots , n$, and this  yields \eqref{diffh}. 
Further, since $h^*(x-1) \le h^*(x)$ and $h(n-x) \le h(n+1-x) $, thanks to Lemma~\ref{hx}, we also obtain 
$0 \le h^*(x) - h^*(x-1) \le m$ for $x=1, \dots , n$.
\end{proof}

\begin{corollary}\label{Mono} 
Let $P=(E,\rho)$ be a $(q,m)$-polymatroid and let $K= \rk P$. 
Then 
$$
d_r(P) < d_{r+m}(P) \quad \text{for any positive integer $r$ such that } r+m \le K.
$$
and
$$
d_s(P^*) < d_{s+m}(P^*) \quad \text{for any positive integer $s$ such that } s+m \le mn-K.
$$
\end{corollary}

\begin{proof}
Let $r$ be a positive integer with $r+m \le K$. 
Then $d_r(P) \le d_{r+m}(P)$, by Proposition~\ref{WeakMono}. 
Suppose, if possible, $d_r(P) = d_{r+m}(P) = x$, say. Then by \eqref{drP} in Lemma~\ref{hx},
$h^*(x-1) < r$ and $h^*(x)\ge r+m$. Consequently,  $h^*(x) - h^*(x-1) \ge m+1$. This contradicts the last assertion in Lemma~\ref{hhstar}. Thus $d_r(P) < d_{r+m}(P)$. Replacing $P$ by $P^*$, we obtain the desired inequality for the generalized weights of $P^*$. 
\end{proof}

%

We shall now proceed to establish a version of Wei duality for the generalized weights of $(q,m)$-polymatroids. Recall that if $C$ is a $[n,k]_q$-code, and $d_1, \dots, d_k$ are the generalized Hamming weights (GHW) of $C$ and $d_1^\perp, \dots , d_{n-k}^\perp$ are the GHW of the dual of $C$, then Wei duality states that the values
$$
n+1 - d_1, \dots , n+1-d_k \quad \text{and} \quad d_1^\perp, \dots , d_{n-k}^\perp
$$
are all distinct and their union is precisely the set $\{1, \dots , n\}$. In the setting of a polymatroid $P=(E, \rho)$ of rank $K$,  the generalized weights of $P$ and its dual $P^*$ lie between 1 and $n$, and we can similarly consider
$$
n+1 - d_1(P), \dots , n+1-d_K(P) \quad \text{and} \quad d_1(P^*), \dots , d_{mn-K}(P^*).
$$
But these  $mn$ values would not constitute $\{1, \dots , mn\}$ when $m\ge 2$, since  they lie between $1$ and $n$.  But one could ask for some ``$m$-fold" version of Wei duality, and that is what we give in the next theorem and the corollary that follows.  These results are 
inspired by the related results of Ravagnani \cite{ravagnani16} and also of Mart\'{\i}nez-Pe\~{n}as and Matsumoto \cite{MM} about generalized weights (GW) and relative GW of Delsarte rank metric codes.   

\begin{theorem} \label{partialWei}
Let $P=(E,\rho)$ be a $(q,m)$-polymatroid of rank $K.$ Also, let  $p,i,j$ be integers such that 
$1\le p+im \le mn-K$ and $1 \le p+K+jm \le K.$
Then 
$$
d_{p+im}(P^*) \ne n+1-d_{p+K+jm}(P).
$$
\end{theorem} 

\begin{proof}
Write $r= p+K+jm$, $s= p+im$, and $x=d_r(P)$. Let  $h$ and $h^*$ be as in Lemma~\ref{hx}. 
In view of \eqref{diffh}, let 
$$
g = h(n+1-x) - h(n-x) = m - \left( h^*(x) - h^*(x-1) \right).
$$
Then using \eqref{drP}, we see that 
$$
r \le h^*(x)  \quad \text{and} \quad r + m - g = h^*(x) + \left( r - h^*(x-1)\right) > h^*(x).
$$
Thus $r \le h^*(x) < r+m-g$, and therefore by \eqref{basich}, we obtain
$$
p+ m(j+n-x) = 
r + m(n-x) - K \le h(n-x)  < r+ m(n-x+1) -g -K. 
$$
The second inequality above implies that 
$$
h(n+1-x) = h(n-x) + g <  r+ m(n-x+1) -K = p+ m(j+n - x+1).
$$
Now suppose, if possible, $n+1 -x = d_s(P^*)$. Then by \eqref{dsPstar}, $h(n-x) < s \le h(n+1-x)$. Combining this with the inequalities obtained earlier, we see that 
$$
p+ m(j+n-x) < s < p+ m(j+n-x)+m.
$$
But this contradicts the fact that $s\equiv p\!\pmod m$. 
\end{proof}

\begin{corollary} \label{important}
Let $P=(E,\rho)$ be a $(q,m)$-polymatroid of rank $K$, and let $s$ be an integer 
such that $0\le s < m$. Define 
\begin{eqnarray*}
&& W_s(P) = \{ d_r(P) : r=1, \dots , \rk P \text{ and } r \equiv s \, (\mathrm{mod} \ m) \}, \text{ and } \\
&& \overline{W}_s(P) = \{n+1- d_r(P) : r=1, \dots , \rk P \text{ and } r \equiv s \, (\mathrm{mod} \ m)\}. 
\end{eqnarray*}
Also define $W_s(P^*)$ and  $\overline{W}_{s}(P^*)$ in a similar manner. 
Then 
$$
W_s(P^*)=\{1,2,\cdots,n\} \setminus \overline{W}_{s+^m K}(P),
$$
where $s+^m K$ means the integer in $\{0,1,\cdots,m-1\}$ congruent to $s+K$ modulo~$m$.
\end{corollary}

\begin{proof}
By Theorem~\ref{partialWei}, the sets $W_s(P^*)$ and  $\overline{W}_{s+^m K}(P)$ are disjoint, and by Proposition~\ref{WeakMono}, they are subsets of $\{1,2,\cdots,n\}$. Thus it suffices to show that the sum of their cardinalities is at least $n$
(and therefore exactly $n$). To this end, write $s+K = Am + B$ for integers $A, B$ with $0\le B < m$. Note that $s+^m K = B$. Let us first consider the case 
$s=0$. Here, by the definition of $W_s(P^*)$, and the frequent leaps, guaranteed by Corollary \ref{Mono}, of the $d_{s+jm}(P^*)$ as $j$ increases, we see that: 
$$
|W_s(P^*)| \ge \left\lfloor \frac{mn-K}{m}\right\rfloor =  \left\lfloor \frac{mn-Am -B}{m}\right\rfloor = \begin{cases} n-A & \text{if } B=0, \\ n - A - 1 & \text{if } B\ge 1. \end{cases}
$$
On the other hand,   Corollary~\ref{Mono} also shows that 
$$
|\overline{W}_B(P)| \ge |W_B(P)|= \begin{cases}  \left\lfloor\frac{K}{m}\right\rfloor = A & \text{if } B=0, \\[0.5em] 
1 +  \left\lfloor\frac{K-B}{m}\right\rfloor = 1+A & \text{if } B\ge 1. \end{cases}
$$
Thus, $|W_s(P^*)| + | \overline{W}_{s+^m K}(P)| \ge n$. 
Consequently, $|W_s(P^*)| + | \overline{W}_{s+^m K}(P)| = n$. 
Next, suppose $s>0$. Here, in a similar manner, Corollary~\ref{Mono} shows that 
$$
|W_s(P^*)| \ge 1+ \left\lfloor \frac{mn-K-s}{m}\right\rfloor =  1+ \left\lfloor \frac{mn-Am -B}{m}\right\rfloor = \begin{cases} 1+n-A & \text{if } B=0, \\ n - A  & \text{if } B\ge 1. \end{cases}
$$
and also that 
$$
|\overline{W}_B(P)| = |W_B(P)| \ge \begin{cases}  \left\lfloor\frac{K-B}{m}\right\rfloor = \left\lfloor\frac{Am-s}{m}\right\rfloor 
= A-1 & \text{if } B=0, \\[0.5em] 
1 +  \left\lfloor\frac{K-B}{m}\right\rfloor =  1+ \left\lfloor\frac{Am-s}{m}\right\rfloor = A & \text{if } B\ge 1. \end{cases}
$$
So, once again $|W_s(P^*)| + | \overline{W}_{s+^m K}(P)| \ge n$ (and hence equal to $n$), 
as desired. 
\end{proof}

\begin{remark} \label{reciprocity}
{\rm
The above corollary shows that the generalized weights of a $(q,m)$-polymatroid $P$ and the generalized weights of its dual $P^*$ determine each other. 
Indeed, first we treat only the $d_r(P^*)$ and $d_{r+^mK}(P)$, for $r \equiv s$, for a fixed value of $s$. By Corollary \ref{important} they determine each other. Since this is true for each fixed s, as $s$ varies in $\{0,1,\cdots,m-1\}$, the assertion holds.

We remark also that the proofs given here of Theorem~\ref{partialWei} and the two preceding lemmas  are motivated by the proofs of the corresponding results for usual matroids (see, e.g., 
\cite[Proposisjon 5.18]{larsen05}). Further, our proof of Corollary~\ref{important} uses arguments that are analogous to those in the proof of  \cite[Corollary 38]{ravagnani16}. 
}
\end{remark}

\section{Generalized  Weights of Delsarte Rank Metric Codes}
\label{sec4}

In this and the subsequent sections, we will denote by $\Mat$, or simply by $\M$ the space of all $m\times n $ matrices with entries in the finite field $\Fq$. Note that $\M$ is a vector space over $\Fq$ of dimension $mn$. 

As stated in the Introduction, by a \emph{Delsarte rank metric code},  or simply a \emph{Delsarte code}, we mean a 
$\Fq$-linear subspace of $\M$.  We write $\dim_{\Fq} C$, or simply $\dim C$, to denote the dimension of a Delsarte code $C$. 

Following Shiromoto \cite{S}, we associate to a Delsarte code $C$, a family of subcodes of $C$ indexed by subspaces of $E=\Fq^n$, and a $(q,m)$-polymatroid as follows. 

\begin{definition}\label{CXrhoC}
{\rm 
Let $C$ be a Delsarte code. 
\begin{enumerate}
\item[{\rm (a)}] Given any $X\in \Sigma (E)$, we define $C(X)$ to be the subspace of $C$ consisting of all matrices in $C$  
whose row space is contained in $X$.
\item[{\rm (b)}]By $\rho^{ }_C$ we denote the function from $\Sigma(E)$ to ${\mathbb{N}_0}$ defined by 
\begin{equation}\label{eq:rhoC}
\rho^{ }_C(X)=\dim_{\mathbb{F}_q}C-\dim_{\mathbb{F}_q}C(X^{\perp}) \quad \text{ for } X \in \Sigma(E).
\end{equation} 
Further, by $P(C)$ we denote the $(q,m)$-polymatroid $(E, \rho^{ }_C)$. 
\end{enumerate}
}
\end{definition}  
                          
\begin{remark}\label{nuCstar}
{\rm
It is shown in \cite[Proposition 3]{S} that $P(C)=(E, \rho^{ }_C)$ does indeed satisfy all the axioms of a $(q,m)$-polymatroid. We note also that the conullity function $\nu^*_{C}$ of $P(C)$ is given by 
\begin{equation}\label{nuC}
\nu^*_C (X) = \dim C(X) \quad \text{for } X \in \Sigma(E).
\end{equation}
}
\end{remark}

\begin{example} \label{MRD}
{\rm 
Assume, for simplicity, that $m > n$. Let $C \subseteq \Mat$ be a MRD code of dimension $K$.  (See, for example, \cite[\S\, 2]{ravagnani15} for the definition and basic facts about MRD codes.) Then $C$ is a Delsarte code 
such that $K = \dim C$ is divisible by $m$ and 
$C(X)=\{0\}$ for all subspaces $X$ of $E$ with $\dim X \le n-\frac{K}{m}$. The latter follows, for instance, from \cite[Proposition 6.2]{GJLR}. This shows that $\rho^{ }_C(X) =K$ if $X\in \Sigma(E)$ with $\dim X \ge K/m$. Further, in view of \cite[Theorem 6.4]{GJLR}, we see that  $\rho^{ }_C(X) =m \dim X$ if $X\in \Sigma(E)$ with $\dim X \le K/m$. 
It follows that $P(C)$ is the uniform matroid $U(\frac{K}{m},n)$. 
}
\end{example}


In general we loosely think of $X$ as (containing) the ``support" of the code $C(X)$, and $C(X)$ as ``the subcode of $C$ supported on $X$".
This gives rise to the idea of the $r$th generalized weight as the smallest dimension of a subspace that can 
support an $r$-dimensional subcode, in analogy with the smallest cardinality of a subset that can 
support an $r$-dimensional subcode, for a block code with Hamming metric. In other words we define:

\begin{definition}\label{Def-drC}
{\rm 
Let $C$ be a Delsarte code and let $K=\dim C$. For $r=1,\cdots,K$, 
the $r$th \emph{generalized  weight} of $C$ is defined by 
$$
d_r(C)=\min \{\dim X : X\in \Sigma(E) \text{ with }   \dim C(X) \ge r 
\}.
$$
}
\end{definition}

\begin{remark}
{\rm This definition is in harmony 
with the one given by 
Mart\'{\i}nez-Pe\~{n}as and Matsumoto  \cite{MM} of the $r$th \emph{relative generalized weight}:
$$
d_{\M,r}(C_1,C_2)=\min \{\dim X : X\in \Sigma(E) \text{ with } 
\dim C_1(X)-\dim C_2(X) \ge r 
\}
$$ 
for pairs of Delsarte codes $C_2 \subseteq C_1$. Indeed, our $d_r(C)$
coincides with their $d_{\M,r}(C,\{0\}).$ See also their Appendix A, where duality theory for a single code is treated.}
%
\end{remark}

We then have defined generalized  weights for $(q,m)$-polymatroids as well as Delsarte codes, in a way intended to
give the following result as a consequence:

\begin{theorem}
Let $C$ be a Delsarte code and let $P(C) = (E, \rho^{ }_C)$ be the corresponding $(q,m)$-polymatroid. Then 
$\dim C = \rk P(C)$ and 
$$
d_r(C)=d_r(P(C)) \quad \text{for } r=1, \dots, \dim C.
$$
\end{theorem}

\begin{proof}  
Clearly, $\rho^{ }_C (E) = \dim C$. Thus $\dim C = \rk P(C)$ and for any $X\in \Sigma(E)$, 
$$
\dim C(X) = \dim C -  \rho^{ }_C (X^\perp) =   \rho^{ }_C (E) -  \rho^{ }_C(X^\perp) = \nu^*(X),
$$
where $\nu^*$ denotes the conullity function of $P(C)$. It follows that
$$
d_r(C) = \min \{\dim X : X\in \Sigma (E) \text{ with }  {\nu}^*(X) \ge r \} = d_r(P(C))
$$
for any $r=1, \dots , \dim C$. 
\end{proof}

\subsection{Duality of Delsarte rank metric codes}

As indicated in the Introduction, the notion of dual for Delsarte reank metric codes is defined using the trace product. 
See for example \cite[Definition 34]{ravagnani16}. We record the basic definition below. 

\begin{definition}\label{def:tracedual}
{\rm 
Let $C$ be a Delsarte code. The  \emph{trace dual}, or simply the  \emph{dual}, of $C$ is the Delsarte code  $C^{\perp}$ defined by 
$$
C^\perp = \{ N \in \Mat : \mathrm{Trace}(MN^t) =0 \text{ for all } M \in C\},
$$
where $N^t$ denotes the transpose of a $m\times n$ matrix $N$ and, as usual, $\mathrm{Trace}(MN^t)$ is the trace of the square matrix $MN^t$, i.e., the sum of all its diagonal entries. 
}
\end{definition}

There is a natural connection between duals of Delsarte codes and the duals of $(q,m)$-polymatroids. It is shown by 
Shiromoto \cite{S} as well as Gorla, Jurrius, Lopez and Ravagnani \cite{GJLR}, and we record it below. 

\begin{theorem} \label{commutes} \cite[Proposition 11]{S}
Let $C$ be a Delsarte code. Then
\[P(C^{\perp})=P(C)^*.\]
\end{theorem}

The proof is quite short and natural and given in \cite[Proposition 11]{S}, and also in \cite[Theorem 8.1]{GJLR}.
Note that this gives another proof of \cite[Proposition 65]{MM}.

\begin{corollary} \label{codewei}
If we know all the generalized weights of a Delsarte rank-metric code $C$, then the generalized  weights of $C^{\perp}$ are uniquely determined (and vice versa). Moreover these ordered sets of generalized  weights determine each other in a way described by Theorem \ref{partialWei}, if one substitutes the codes in question with their respective polymatroids.
\end{corollary}

\begin{proof}
Follows from Theorem \ref{commutes} and Remark \ref{reciprocity}.
\end{proof}

\begin{example}
{\rm Assume that $m>n$. 
Let $C$ be an MRD code of dimension $K=mK'$. As we saw in Example \ref{MRD}, $P(C)=U(K',n)$.
Then $P(C^{\perp})=P(C)^*=U(n-K',n)$, which translated back to the language of codes gives $\dim C(X)=\{0\}$, for all subspaces 
$X$ of $E$ of dimension at most $K'$. Hence we obtain the well-known fact that if $C$ is an MRD code of dimension $K=mK'$, then $C^{\perp}$ is an MRD-code of dimension $mn-K=m(n-K').$}
\end{example}

\begin{example}
{\rm If  $C$ is a Delsarte code with $\dim C$ is divisible by $m$, as for example for  a Gabidulin code, or an MRD-code, then the conclusion of Corollary \ref{important} is  
$$
W_s(P^*)=\{1,2,\cdots,n\} \setminus \overline{W}_{s}(P) \quad \text{for each $s=1,\dots,m.$}
$$
Consequently, a Delsarte code $C$ with dimension divisible by $m$ is MRD if and only if  
$d_s(C)= n-K'+1$ for each $s=1,\dots,m$  (and then automatically $d_{s+im}(C)=n-K'+i+2$ for  $i=1,\dots, K'-1,$ while $d_s(C^{\perp})=K'+1$ and $d_{s+jm}(C^{\perp})=K'+j+2$ for $j=1,\dots,n-K'-1).$ 
}
\end{example}

\subsection{Another definition of generalized weights}
Ravagnani has given another definition in \cite[Definition 23]{ravagnani16} of generalized weights of Delsarte codes that is based on the following notion of anticodes. 

\begin{definition}
{\rm
By an \emph{optimal anticode} we mean an $\mathbb{F}_q$-linear subspace ${\mathcal{A}}$ of $\Mat$ such that 
$\dim_{\mathbb{F}_q} {\mathcal{A}} =m \left(\mathrm{maxrank}({\mathcal{A}})\right)$, where $\mathrm{maxrank}({\mathcal{A}})$ denotes 
the maximum possible rank of any matrix in ${\mathcal{A}}$.
}
\end{definition}

Here is Ravagnani's definition of generalized weights of Delsarte codes.

\begin{definition} \label{acodedef}
{\rm 
Let $C$ be a Delsarte code of dimension $K$. 
Define
 $$
a_r(C)=\frac{1}{m}\min \left\{\dim_{\mathbb{F}_q}{\mathcal{A}} : {\mathcal{A}}  \textrm{ an optimal anticode such that  }\dim_{\mathbb{F}_q}({\mathcal{A}} \cap C)\ge r\right\}.
$$
}
\end{definition}

A relationship between the two notions (given in Definitions \ref{Def-drC} and \ref{acodedef}) is stated below. This result is given in \cite[Theorem 9]{MM} as well as \cite[Proposition 2.11]{GJLR}, and we refer to the former for a proof of the following theorem. 

\begin{theorem}\label{ardr}
Let $C$ be a Delsarte code. 
Then for each $r=1, \dots , \dim C$, 
$$
a_r(C)=d_r(C) \text{ if } \textcolor{blue}{m > n}, \quad \text{whereas} \quad a_r(C) \le d_r(C) \text{ if } m=n.
$$
%
\end{theorem}

\begin{remark} \label{imprem}
{\rm Theorem \ref{ardr} gives a second proof of Corollary \ref{codewei} when \textcolor{blue}{ $m > n$}, since the corresponding result for the $a_r(C)$ and $a_r(C^{\perp})$ is given by Ravagnani \cite[Corollary 38]{ravagnani16}. Also, when $m=n$, both  \cite[Corollary 38]{ravagnani16} and Corollary \ref{codewei} are still valid, but the $a_r$ and the $d_r$ are not necessarily the same. An example where they are different is given by Mart\'{\i}nez-Pe\~{n}as and Matsumoto  \cite[Section IX,C]{MM}. Another proof of Corollary \ref{codewei} is given in 
 \cite[Lemma 66, Corollary 68]{MM}. 
We refer to \cite[Theorem 5.14]{G} and \cite[Theorem 5.18]{G} for a fuller treatment of the case $m=n$,  and also for the cases where $m < n$. We will, however, based on that treatment, return to the case $m=n$ in Subsection \ref{square}. Here we just remark briefly that if, on the other hand, $m<n$, then it will  be more natural to work with the $(q,n)$-polymatroid $P(C^T)$ (where $C^T$ is the set of  transposes of matrices in $C$) than with the $(q,m)$-polymatroid $P(C).$ In particular it follows from \cite[Theorem 5.18]{G} that the generalized weights given in \cite{ravagnani16} coincide with the generalized weights for the $(q,n)$-polymatroid $P(C^T)$. 
Hence Wei duality for polymatroids gives a second proof for the Wei duality of Ravagnani's generalized weight also for $m<n $.}
\end{remark}
\begin{remark}
{\rm It is a main point in our exposition that we can prove our main results, Theorem \ref{partialWei}, and Corollary \ref{important},
without even mentioning Delsarte rank metric codes, but at the same time, these results imply the 
``Wei duality" when the $(q,m)$-polymatroid in question indeed comes from a Delsarte rank metric code. 
One might wonder whether there are $(q,m)$-polymatroids that do not come from Delsarte rank metric codes, but where our Wei duality may give interesting descriptions for other objects. For usual matroids, there are matroids that are non-representable, and thus do not come from linear codes. An example is the non-Pappus matroid $\mathcal{M}$ (say), with ground set 
of cardinality $9$. 
The Wei duality of matroids, as described for example in \cite{BJMS} or \cite[Proposisjon 5.18]{larsen05} without mentioning codes, is enjoyed by $\mathcal{M}$ as well.  But for the non-Pappus matroid $\mathcal{M}$, 
Wei duality can also be interpreted  in a coding theoretic sense, except that instead of (linear block) codes, we have to consider the so called almost affine codes, which can be nonlinear and whose alphabet set need not even be a field; see
\cite[Example 1]{JV}. 
In analogy with this,  we may ask the following. 
Is there 
a  class of codes strictly bigger (or quite different) than that of Delsarte rank metric codes, such that 
the codes in this class 
give rise to $(q,m)$-polymatroids, and where  duality of codes corresponds to duality of $(q,m)$-polymatroids, and moreover, ``Wei duality" for $(q,m)$-polymatroids can be interpreted in a  coding theoretic sense?
}
\end{remark}

%
%

\begin{example} \label{sumsofmatroids}
{\rm 
We will describe a $(q,m)$-polymatroid $P$, which is not defined as a $P(C)$ for a single Delsarte 
code $C$, but is derived from a collection of $m$ codes.

 Let $C_1,\dots,C_m$ be $m$ block codes of length $n$ over $\mathbb{F}_q$.
We will view the codewords as $1\times n$ matrices, and the $C_i$ as Delsarte rank metric codes. Let 
$$
K_i =\dim C_i  \, \text{ for } 1=\dots,m,  \quad \text{and let} \quad K=\sum_{i=1}^m K_i.
$$
Note that for each $i=1, \dots , m$, the space $C_i(X)$ coincides with $C_i\cap X$ whenever $X\in \Sigma(E)$ and thus the rank function $\rho_i = \rho^{ }_{C_i}$ of $P(C_i)$ is given by 
$$
\rho_i(X)=\dim C_i - \dim (C_i\cap X^{\perp}) \quad \text{for } X \in \Sigma(E).
$$
Also note that the  trace dual $C_i^{\perp}$ is simply the usual orthogonal complement of $C_i$ as a block code. 
Clearly, $P(C_i)$ as well as $P(C_i^{\perp})$ are $(q,1)$-polymatroids. Since $P(C_i)^*=P(C_i^{\perp})$, the rank function $\rho^*_i$ of $P(C_i)^*$ satisfies 
$$
\rho_i^*(X)=\rho_i(X^{\perp})+\dim X-\rho_i(E) = \dim C_i^{\perp} - \dim (C_i^{\perp}\cap X^{\perp})  
\quad \text{for } X \in \Sigma(E).
$$
Now define $\rho: \Sigma(E) \to \mathbb{N}_0$ by 
%
$$
\rho(X)=\rho_1(X)+\cdots+\rho_m(X) \quad \text{for } X \in \Sigma(E).
$$
It is easy to see that $P=(E,\rho)$ satisfies all the axioms for $(q,m)$-polymatroids. Moreover, as a $(q,m)$-polymatroid, for any $X \in \Sigma(E)$, we obtain 
$$
\rho^*(X)=\rho(X^{\perp})+m\dim X-\rho(E)=\sum_{i=1}^m\rho_i(X^{\perp})+\sum_{i=1}^m\dim X-\sum_{i=1}^m\rho_i(E)= \sum_{i=1}^m \rho_i^*(X) 
$$
and hence if $\nu^*$ denotes the conullity function of $P$, then, in view of \eqref{nuC}, we find 
$$
{\nu}^*(X)=m\dim X -\rho^*(X)=\sum_{i=1}^m \dim X-\rho_i^*(X) = 
\sum_{i=1}^m\dim C_i(X)=\sum_{i=1}^m\dim C_i\cap X.
$$
Thus the generalized weights of the $(q,m)$-polymatroid $P$  are given by 
$$
d_r(P)=\min \big\{\dim X : X \in \Sigma(E) \text{ with } \sum_{i=1}^m \dim (C_i\cap X) \ge r\big\} 
$$
for $ r=1,\dots, K$, whereas the generalized weights of $P^*$ are given by 
  $$
d_r(P^*)=\min \big\{\dim X : X \in \Sigma(E) \text{ with } \sum_{i=1}^m \dim (C_i^{\perp}\cap X) \ge r\big\} 
$$
for $r=1,\dots, mn - K$.
A relation between these two sets of generalized weights is described in Theorem \ref{partialWei} and Corollary \ref{important}. From the construction of $P$, we see that $d_1(P)=1,$ unless all $C_i$ are zero, since we may take some one dimensional $X$ contained in some $C_i$, and calculate $\dim C_i \cap X =1$. Analogously, $d_1(P^*)=1$ as well, unless $C_i=E$, for all $i$. On the other hand $d_K(P)=n$, unless there is a strict subspace of $E$ that contains all the codes $C_i$. So, if   $C_i \ne E$, for some $i$, but the span of 
$C_1\cup \dots \cup C_m$  is $E$, then $d_1(P^*)=1$ and $d_{\rk(P)}(P)=n$, a possibilty excluded if the ``usual" Wei duality were applicable, 
but which 
indeed may occur, and in fact does occur, under the ``revised" duality  described in Theorem \ref{partialWei} and 
Corollary \ref{important}. Given that $d_1(P^*)=1$,  Theorem \ref{partialWei} only prohibits that $d_r(P)=n$ for all $r$
congruent to $K+1$ modulo $m$. But $K$ is certainly not congruent to $K+1$ modulo $m$, and so $d_K(P)$ may very well be $n$.

In a certain sense  this example is also associated to $m\times n$ matrices, since each element of $C_1\times \cdots \times C_m$ could be presented as $m$ codewords of length $n$ arranged as an $m \times n$ matrix. Our function $\rho$ does however not ``measure the behaviour" of the row space of the matrix, including all \emph{linear combinations of} the rows, as a rank function of a $P(C)$ of a Delsarte rank metric code $C$ would have done; it only ``measures the behaviour" of the individual rows.
``Intermediate" examples could have been made by taking $\rho(X)=\rho_1(X)+\cdots+\rho_s(X),$   where $\rho_i$ is the rank function of an $m_i \times n$ Delsarte rank metric code,  for $1\le i \le s$, and  where $m_1, \dots , m_s$ are nonnegative integers with 
$m_1+ \cdots + m_s = m$. 
 }
\end{example}    

\section{Demi-polymatroids and their Generalized Weights} 
\label{sec5}                       

In this section, we discuss a generalization of the notion of $(q,m)$-polymatroids, and observe that most of the results in Section \ref{sec:defnnot} can be extended in a more general context. In the next section we shall see how this  generalization is relevant for Delsarte rank metric codes.

\begin{definition} \label{someR}
{\rm
A $(q,m)$-{\em demi-polymatroid} is an ordered pair $(E,\rho)$ consisting of
the vector space $E=\mathbb{F}_q^n$ over $\mathbb{F}_q$ and a function $\rho:\Sigma(E)\to\mathbb{N}_0$
satisfying the following three conditions: 
\begin{list}{}{\leftmargin=1em\topsep=1.5mm\itemsep=1mm}
\item[{\rm (R1)}]  $0\leq \rho(X) \leq m\dim X$ for all $X\in \Sigma(E)$;
\item[{\rm (R2)}]  $\rho(X) \le \rho(Y)$ for all $X, Y \in \Sigma(E)$ with $X \subseteq Y$;
\item[{\rm (R4)}]  $\rho^*:\Sigma(E)\to\mathbb{N}_0$ defined by $\rho^*(X)=\rho(X^{\perp})+m\dim X-\rho(E)$ for $X\in \Sigma(E)$  
also satisfies (R1) and (R2).
\end{list}
}
\end{definition}

The notion of dual is defined exactly as in the case of $(q,m)$-polymatroids and we have the following analogue of Proposition \ref{dualdef} or equivalently,  \cite[Proposition 5]{S}.

\begin{proposition} \label{dualdef2} 
Let $P=(E,\rho)$ be a $(q,m)$-demi-polymatroid and let $\rho^*$ be as in $\mathrm{(R4)}$ above. Then the ordered pair $(E,\rho^*)$ 
is also a $(q,m)$-demi-polymatroid, denoted by $P^*$ and 
called the dual $(q,m)$-demi-polymatroid of $P$.
\end{proposition} 

\begin{proof}
By the definition of a $(q,m)$-demi-polymatroid, $(E,\rho^*)$ satisfies (R1) and (R2). Moreover, since
$(\rho^*)^*(X)=\rho^*(X^{\perp})+m\dim X-\rho^*(E)$ equals 
$$
 \big(\rho(X)+m\dim X^{\perp}-\rho(E)\big)+m\dim X - \big( \rho(\{\mathbf{0}\})+m\dim E-\rho(E)\big)
$$
and since $\rho(\{\mathbf{0}\})=0$, thanks to  (R1), and $\dim X + \dim X^\perp = \dim E$, we see that $(\rho^*)^*=\rho$, and so (R4) is satisfied by $(E, \rho^*)$. 
\end{proof}

\begin{proposition}\label{BasicDemi}
Let $P=(E, \rho)$ be a $(q,m)$-polymatroid. Then $P$ is a $(q,m)$-demi-polymatroid. Moreover, if $\nu$ and $\nu^*$ denote, as usual, the nullity and conullity functions of $P$, then both $(E, \nu)$ and $(E, \nu^*)$ are 
$(q,m)$-demi-polymatroids. In particular, there exists a $(q,m)$-demi-polymatroid that is not a $(q,m)$-polymatroid.
\end{proposition}

\begin{proof}
The first assertion follows from Proposition \ref{dualdef}. 
Next, recall that 
\begin{equation}\label{nunustar}
{\nu}(X)=m\dim X - \rho(X) \quad \text{and} \quad 
{\nu}^*(X)=m\dim X -\rho^*(X)= \rho(E)-\rho(X^{\perp})
\end{equation}
for any $X\in \Sigma (E)$. Note, in particular, 
that $\nu(\{\mathbf{0}\})=0 = \nu^*(\{\mathbf{0}\})$, and so 
Proposition \ref{tech} implies that both $\nu$ and $\nu^*$ satisfy (R1) and (R2). 
The dual of $\nu$ in the sense of (R4) is  the function that associates to every $X\in \Sigma (E)$ the integer
$$
\nu(X^{\perp})+m\dim X-\nu(E) = \big( m\dim X^{\perp}- \rho(X^{\perp})\big) +m\dim X -\big( mn - \rho(E) \big), 
$$
which is easily seen to be ${\nu}^*(X)$. 
Thus the two possible meanings of $\nu^*$ coincide. Hence by Proposition \ref{tech}, $(E, \nu)$ satisfies (R4) as well. Furthermore, it is readily seen that $(\nu^*)^*=\nu$, and so Proposition \ref{tech} also shows that $(E, \nu^*)$ satisfies (R4). Thus both $(E, \nu)$ and $(E, \nu^*)$ are 
$(q,m)$-demi-polymatroids. In particular, if we take $P = U(1,2)$, then from Remark 
\ref{exa:not_polymatroid}, we see that the corresponding pair $(E, \nu)$ is a $(q,m)$-demi-polymatroid, but not a $(q,m)$-polymatroid.
%
\end{proof}

\begin{example}\label{nuCX}
{\rm 
If $C\subseteq \Mat$ is a Delsarte code and $\delta: \Sigma(E)\to\mathbb{N}_0$ is defined~by 
$$
\delta (X) = \dim C(X) \quad \text{for } X\in \Sigma(E),  
$$
then Remark \ref{nuCstar} and Proposition \ref{BasicDemi} shows that $(E, \delta)$ is a $(q,m)$-demi-polymatroid. Moreover, in view of Example \ref{MRD} and  Remark 
\ref{exa:not_polymatroid}, we see that $(E, \delta)$ is, in general, not a $(q,m)$-polymatroid.
}
\end{example}

In general, if $P=(E, \rho)$ is a $(q,m)$-demi-polymatroid, then the nullity function $\nu$ and  the conullity function $\nu^*$ of $P$ are defined in exactly the same way as in the case of  $(q,m)$-polymatroids, i.e., by equation \eqref{nunustar}. 
Our proof of Proposition \ref{tech}\,(a) used the property (R3), which is not available in the case of $(q,m)$-demi-polymatroids, but we will 
show below that the result is still valid in this case. 

\begin{proposition} \label{tech2}
Let $P= (E, \rho)$ be a $(q,m)$-demi-polymatroid and let $X,Y \in \Sigma (E)$ with $X\subseteq Y$.  Then: 
\begin{enumerate}
\item[{\rm (a)}]  ${\nu}(X) \le {\nu}(Y)$ and 
$\, {\nu}^*(X) \le {\nu}^*(Y)$; 
\item[{\rm (b)}] ${\nu}(Y)-{\nu}(X)  \le m(\dim Y - \dim X)$ and 
$\, {\nu}^*(Y)-{\nu}^*(X)  \le m(\dim Y - \dim X)$.
\end{enumerate}
\end{proposition}

\begin{proof}
(a) 
Since $\rho^*$ satisfies (R2), thanks to (R4), and since $Y^\perp \subseteq X^\perp$, we see that
$\rho^*(Y^\perp) \le \rho^*(X^\perp)$, which shows that $\rho(Y) + m \dim Y^\perp \le \rho(X) + m \dim X^\perp$. Subtracting from $mn =m\dim E$, we find 
${\nu}(X) \le {\nu}(Y)$. Similarly, $\, {\nu}^*(X) \le {\nu}^*(Y)$.
%

\smallskip

(b) The proof of Proposition \ref{tech} (b) only uses (R2) for $\rho$ and $\rho^*$, and so it is still valid here.
\end{proof}

We define the generalized weights for $(q,m)$-demi-polymatroids in exactly the same way as in the case of $(q,m)$-polymatroids:

\begin{definition}\label{Defdrdemi}
Let $P=(E,\rho)$ be a  $(q,m)$-demi-polymatroid. 
For $r=1,\dots,\rho(E)$, the 
\emph{$r$th generalized  weight} of $P$ is 
defined by 
$$
d_r(P)= \min \{\dim X : X\in \Sigma(E) \text{ with }  \nu^*(X) \ge r\}.
$$ 
\end{definition}

We then have the following more general result about Wei-type duality. 

\begin{theorem}\label{DemiDuality}
The results in Theorem \ref{partialWei} and Corollary \ref{important} are valid also for $(q,m)$-demi-polymatroids $P$.
\end{theorem}

\begin{proof}
Examining the proof of Theorem  \ref{partialWei}, we see that all arguments follows from axioms (R1), (R2), (R4), and there is no need for axiom (R3). One does use Proposition \ref{tech} whose proof depended on (R3), but we have established it for 
$(q,m)$-demi-polymatroids in Proposition \ref{tech2} above. 
\end{proof}

\subsection{Wei duality for square matrices} \label{square}
In this subsection, we consider the case when $m=n$. In this case, if $C \subseteq \Mat$ is a Delsarte rank metric code, then so is $C^T:=\{M^t: M \in C\}$, and thus, we obtain two $(q,m)$-polymatroids $P(C) = (E, \rho^{ }_C)$ and $P(C^T) = (E, \rho^{ }_{C^T})$, where $E=\mathbb{F}_q^n$ and $\rho^{ }_{C}, \rho^{ }_{C^T}$ are 
as in \eqref{eq:rhoC}.

\begin{proposition} \label{demisquare}
Assume that $m=n$. Let $C \subseteq \Mat$ be a Delsarte rank metric code. Consider $E=\mathbb{F}_q^n$ and define $\rho:\Sigma(E)\to \mathbb{N}_0$ by
$$
\rho(X)=\min \{\rho^{ }_{C}(X), \; \rho^{ }_{C^T}(X)\} \quad \text{for } X\in \Sigma(E).
$$ 
Then 
$P=(E,\rho)$ is a $(q,m)$-demi-polymatroid and its conullity function is given by 
$$
\nu^*(X)=\max \{\dim C(X), \; \dim C^T (X)\} \quad \text{for } X\in \Sigma(E).
$$ 
Moreover, the generalized weights of $P$ are given by
$$
d_r(P)=\min \{d_r(P(C)), \; d_r(P\left( C^T \right) \}
\quad \text{for } r=1, \dots ,  \rho(E).
$$ 
Consequetly, the Wei duality holds for Ravagnani's generalized weights $a_r(C)$. 
\end{proposition}

\begin{proof}
It is obvious that $\rho$ satisfies (R1) and (R2) of Definition \ref{someR}, since we know that each of $\rho_C$ and $\rho_{C^T}$ satisfies these properties. So, in order to prove that $P$ is a $(q,m)$-demi-polymatroid, it remains to show that (R4) is satisfied, 
which means that $\rho^*$ satisfies (R1) and (R2). To this end, let $X\in \Sigma(E)$. Then 
\begin{eqnarray}\label{rhostarmax}
\nonumber \rho^*(X) &= & \rho(X^{\perp}) +m\dim X-\rho(E) \\ \nonumber
&= &  \min \{ \rho^{ }_{C}(X^{\perp}), \; \rho^{ }_{C^T}(X^{\perp}) \}  +m\dim X- \dim C \\ \nonumber
&= & \min \{ \dim C - \dim C(X), \;  \dim C - \dim C^T (X) \}+m \dim X -\dim C \\ 
&= &  m\dim X - \max \{\dim C(X), \; \dim C^T(X)\}
\end{eqnarray}
This implies that $\rho^*(X) \le m\dim X$. Moreover, it also implies that $\rho^*(X) \ge 0$, because from 
\eqref{nuC} and Proposition~\ref{BasicDemi} we see that both 
$m\dim X - \dim C(X)$ and $m\dim X - \dim C^T(X)$ are 
nonnegative. Thus $\rho^*$ satisfies (R1). 
Next, we show that $\rho^*$ satisfies (R2). Fix  $X, Y\in \Sigma(E)$ with $X \subseteq Y$. In view of \eqref{rhostarmax}, the difference $\rho^*(Y)-\rho^*(X)$ can be written as 
$$
m\large(\dim Y - \dim X \large) - \big(\max\{\dim C(Y),\dim C^T(Y)\} - \max\{\dim C(X),\dim C^T(X)\} \big).
$$
Since the expression above is symmetric in $C$ and $C^T$, we may assume without loss of generality that 
$\max\{\dim C(Y),\dim C^T(Y)\}  = \dim C(Y)$. %
Now, 
in case $\max\{\dim C(X),\dim C^T(X)\} = \dim C(X)$, we see that
$$
\rho^*(Y)-\rho^*(X) = 
m\large(\dim Y - \dim X \large) - \large(\dim C(Y) - \dim C(X) \large) = \rho_{C}^*(Y)-\rho_{C}^*(X),
$$
which is nonnegative since $\rho_{C}^*$ satisfies (R1), thanks to Proposition \ref{dualdef}. 
In case $\max\{\dim C(X),\dim C^T(X)\} = \dim C^T(X)$, then $ \dim C^T(X) \ge  \dim C(X)$, and so
$$
\rho^*(Y)-\rho^*(X) = 
m\large(\dim Y - \dim X \large) - \large(\dim C(Y) - \dim C^T(X) \large) \ge \rho_{C}^*(Y)-\rho_{C}^*(X), 
$$
which is again nonnegative. Thus $\rho^*$ satisfies (R2). This proves that $P=(E,\rho)$ is a $(q,m)$-demi-polymatroid. The desired formula for the conullity function of $P$ is immediate from \eqref{rhostarmax}. 
This, in turn, shows that 
$$
d_r(P)=\min \{d_r(P(C)), \; d_r(P\left( C^T \right) \} \quad \text{for } r=1, \dots ,  \rho(E).
$$
Indeed, the inequality $d_r(P) \le \min \{d_r(P(C)), \; d_r(P\left( C^T \right) \}$ is clear from the definition and equation 
\eqref{nuC}. For the other inequality, it suffices to consider $X_0\in \Sigma(E)$ with $\max  \{\dim C(X_0), \; \dim C^T (X_0)\}  \ge r$ such that $d_r(P) = \dim X_0$.

The last assertion about Wei duality for Ravagnani's generalized weights $a_r(C)$
is an immediate consequence of Theorem \ref{DemiDuality} because we know from \cite[Theorem 38]{G}  that $a_r(C) = \min \{d_r(P(C)), \; d_r(P\left( C^T \right) \}$ for $1\le r \le \dim C = \rho(E)$. 
\end{proof}

\section{Flags of Delsarte Rank Metric Codes} 
\label{sec5plus}

\begin{definition}\label{RhoF}
{\rm
By a \emph{flag} of Delsarte codes we shall mean a tuple $\mathsf{F} = (C_1, \dots , C_s)$ of subspaces of $\M = \Mat$ such that $C_s \subseteq C_{s-1} \subseteq \cdots \subseteq C_1$. The \emph{rank function} associated to a flag $\mathsf{F} = (C_1, \dots , C_s)$ 
is the map $\rhoF : \Sigma(E) \to \mathbb{Z}$ given by 
\begin{equation}\label{rhoFdefn}
\rhoF(X) = \sum_{i=1}^s(-1)^{i+1}\dim \rho_i(X), \quad \text{where}\quad 
\rho_i(X)=\dim_{\mathbb{F}_q}C_i-\dim_{\mathbb{F}_q}C_i(X^{\perp})
\end{equation}
for $i=1, \dots , s$ and $X\in \Sigma(E)$. 
}
\end{definition}

Observe that if $\F$ is the singleton flag $(C)$, then $\rhoF$ coincides with the map $\rho^{ }_C$ introduced in Definition~\ref{CXrhoC}. We have noted in Remark~\ref{nuCstar} that $P(C) = (E, \rho^{ }_C)$ is a $(q,m)$-polymatroid. We will show that $\PF = (E, \rhoF)$ is a $(q,m)$-demi-polymatroid for any flag $\F$ of Delsarte codes. The main components of the proof will be shown in the form of a couple of lemmas. 

\begin{lemma}\label{Lem1F}
Let $C_1, C_2$ be Delsarte codes in $\M = \Mat$ such that $C_2\subseteq C_1$ and let $\rho_i =  \rho^{ }_{C_i}$ for $i=1, 2$. Then 
$\rho_2(X) \le \rho_1(X)$ for all $X\in \Sigma(E)$. 
\end{lemma}

\begin{proof}
Note that the row space of any $A\in \M$ consists of vectors $\mathbf{v}A$ as $\mathbf{v}$ varies over $\Fq^m$ (elements of $\Fq^m$ and $\Fq^n$ are thought of as row vectors); also note that 
$(\mathbf{v}A) \cdot \mathbf{u} = \mathbf{u}(\mathbf{v}A)^t = (\mathbf{u}A)^t\mathbf{v}^t$ for any $\mathbf{u}\in \Fq^n$.  Now let $X\in \Sigma(E)$ and define 
$$
U = \{A \in \Mat : \mathbf{u}A^t = \mathbf{0} \text{ for all } \mathbf{u} \in X \}.
$$ 
Clearly, $U$ is a subspace of $\M$ and $C(X^\perp) = C\cap U$ for any Delsarte code $C$. Also, 
$$
\frac{C_2}{C_2\cap U} \simeq \frac{C_2+U}{U} \subseteq \frac{C_1+U}{U} \simeq \frac{C_1}{C_1\cap U}.
$$
Hence $\dim C_2 - \dim C_2 \cap U \le \dim C_1 - \dim C_1 \cap U$, which yields  $\rho_2(X) \le \rho_1(X)$. 
\end{proof}

\begin{lemma}\label{Lem2F}
Let $C_1, C_2$ be Delsarte codes in $\M = \Mat$ such that $C_2\subseteq C_1$ and let $X, Y\in \Sigma(E)$ be such that $X\subseteq Y$. Then 
$$
\dim C_1(X) -  \dim C_2(X) \le \dim C_1(Y) - \dim C_2(Y).
$$
\end{lemma}

\begin{proof}
First observe that $\dim \M(X) = m \dim X$. Indeed, any $A\in \M(X)$ uniquely determines a $\Fq$-linear map $\phi_A: \Fq^m \to X$ given by $\mathbf{v}\mapsto \mathbf{v}A$, and $A\mapsto \phi_A$ defines an isomorphism of $\M(X)$ with the space of all inear maps from $\Fq^m$ to $X$. Similarly, $\dim \M(Y) = m \dim Y$. Let  $\mathcal{B}_X, \mathcal{B}_Y$ and $\mathcal{B}$ be bases of $\M(X), \M(Y)$ and $\M$ such that $\mathcal{B}_X\subseteq \mathcal{B}_Y \subseteq \mathcal{B}$. We can use the basis $\mathcal{B}$ to define a $\Fq$-linear isomorphism $\pi: \M \to \Fq^{mn}$ so that Delsarte codes $C$ in $\M$ can be identified with linear block codes $\pi (C)$ of length $mn$. Write generator matrices of $\pi(C_1)$  and $\pi(C_2)$ 
as 
$$
G_1 = \begin{pmatrix} P & Q & R \\  S & T & U  \end{pmatrix}_{\dim C_1 \times mn} \quad \text{and} \quad G_2=  \begin{pmatrix} P  & Q & R   \end{pmatrix}_{\dim C_2 \times mn} 
$$
where the blocks $P$ and $S$ correspond to coordinates with respect to  $\mathcal{B}_X$ while the blocks $Q$ and $T$ correspond to coordinates with respect to $\mathcal{B}_Y\setminus \mathcal{B}_X$. 
By removing from $G_1$ superfluous rows that may have become linearly dependent when restricted to coordinates w.r.t. $\mathcal{B}_Y$, we see that a generator matrix for $\pi(C_1(Y))$ is of the form 
$$
\begin{pmatrix} P' & Q'  \\  S' & T'  \end{pmatrix} 
$$
and its submatrix $\begin{pmatrix} P' & Q'  \end{pmatrix}$ is a generator matrix for $\pi(C_1(X))$.
Consequently, 
$$
\dim C_1(Y) - \dim C_2(Y) = \rk \begin{pmatrix} P' & Q'  \\  S' & T'  \end{pmatrix} - \rk \begin{pmatrix} P' & Q'  \end{pmatrix} = \rk \begin{pmatrix}  S' & T'  \end{pmatrix}.
$$
On the other hand, 
$$
 \rk \begin{pmatrix}  S' & T'  \end{pmatrix} \ge 
\rk S' \ge \rk  \begin{pmatrix}  P' \\ S'  \end{pmatrix} -  \rk P' =  \dim C_1(X) -  \dim C_2(X).
$$
This proves the desired inequality. 
\end{proof}

Here is the result that was alluded to earlier in this section. 

\begin{theorem} \label{flagrank}
Let $\mathsf{F} = (C_1, \dots , C_s)$ be a flag of Delsarte 
codes in $\M$ and let $\rhoF$ be the rank function associated to $\F$. Then $\PF =(E,\rhoF)$  is a $(q,m)$-demi-polymatroid.
\end{theorem}

\begin{proof}
First, suppose $s$ is even, say $s=2t$. Then for any $X\in \Sigma(E)$, 
\begin{equation}\label{evenrho}
\rhoF(X) = \sum_{i=1}^t \big(\dim \rho^{ }_{2i-1}(X) - \dim \rho^{ }_{2i}(X)\big).
\end{equation}
By Lemma \ref{Lem1F}, each 
summand is nonnegative, and so $\rhoF(X) \ge 0$. In case $s=2t+1$, 
\begin{equation}\label{oddrho}
\rhoF(X) =  \dim \rho^{ }_{2t+1}(X) + \sum_{i=1}^t\big( \dim \rho^{ }_{2i-1}(X) - \dim \rho^{ }_{2i}(X)\big).
\end{equation}
and once again $\rhoF(X) \ge 0$, thanks to Remark~\ref{nuCstar}  and Lemma \ref{Lem1F}. Next, if $s>1$ and if $\F' =(C_2, \dots , C_s)$ denotes the flag obtained from $\F$ by dropping the first term, then by what is just shown $\rhoFp(X) \ge 0$ for any $X\in \Sigma(E)$. Hence 
$$
\rhoF(X) = \rho^{ }_1(X) - \rhoFp(X) \le  \rho^{ }_1(X)  \le m \dim X \quad \text{for all } X\in \Sigma(E).
$$
This shows that $\PF$ satisfies (R1).  Next, let $X, Y\in \Sigma(E)$ with $X \subseteq Y$. We will show that $\rhoF(X) \le \rhoF(Y)$. To this end, observe that  since $\rhoF(X)$ (and likewise $\rhoF(Y)$) can be expressed as in \eqref{evenrho} or \eqref{oddrho}, and since $\rho^{ }_{2t+1}$ satisfies (R2), it suffices to show that the difference
$$
\big( \dim \rho_{1}(Y) - \dim \rho_{2}(Y) \big) - \big(\dim \rho_{1}(X) - \dim \rho_{2}(X) \big)
$$
is nonnegative. But an easy calculation shows that this difference is equal to 
$$
\big( \dim C_1(X^{\perp})-\dim C_2(X^{\perp}) \big)- \big(\dim C_1(Y^{\perp})-\dim C_2(Y^{\perp}) \big).
$$
But since $Y^\perp \subseteq X^\perp$, by Lemma \ref{Lem2F}, the above difference is nonnegative. This shows that $\PF$ satisfies (R2). To prove that $\PF$ satisfies (R4), note that the case $s=1$ is trivial. Thus suppose $s>1$ and let  $\F' =(C_2, \dots , C_s)$. Also, let 
$$
\rho^*(X)= \rhoF(X^{\perp})+m\dim X-\rhoF(E) \quad \text{and}\quad \rho'(X) = \rhoFp(X) \quad \text{for } X\in \Sigma(E).
$$
Since $\rhoF = \rho^{ }_1 - \rho'$ and since $\rho^{*}_1$ satisfies (R1) while $\rho'$ satisfies (R2), we see that
$$
\rho^*(X)= \big( \rho_1(X^{\perp})+m\dim X -\rho_1(E) \big) +\big( \rho'(E)-\rho'(X^{\perp})\big) \ge 0 +0=0
$$
Also, $\rho^*(X)= m\dim X+\big( \rhoF (X^{\perp})-\rhoF (E) \big) \le m\dim X$, since $\rhoF$ satisfies (R2). Thus $\rho^*$ satisfies (R1). Finally, if $X, Y\in \Sigma(E)$ with $X \subseteq Y$, then we can write $\rho^*(Y)-\rho^*(X)
= \big(\rho(Y^{\perp})+m\dim Y - \rho(E)\big)- \big(\rho(X^{\perp})+m\dim X - \rho(E)\big)$ as 
\begin{eqnarray*}
& & m\big(\dim Y -\dim X\big)+\big(\rho(Y^{\perp})-\rho(X^{\perp})\big) \\
 &=& m\big(\dim X^{\perp}-\dim Y^{\perp})-\big( \rho^{ }_1 (X^{\perp})-\ \rho^{ }_1 (Y^{\perp}) \big) +\big(\rho'(X^{\perp})-\rho'(Y^{\perp})\big) \\
 &=&\big( \nu_1(X^{\perp})-\nu_1(Y^{\perp}) \big) +\big(\rho'(X^{\perp})-\rho'(Y^{\perp}) \big),
\end{eqnarray*}
where $\nu_1$ denotes the nullity function of $(E,  \rho^{ }_1 )$. Thus, using Proposition \ref{tech} (a) and the fact that $\rho'$ satisfies (R2), we see that $\rho^*$ satisfies (R4). 
\end{proof}

Using Theorem \ref{flagrank} and Definition \ref{Defdrdemi}, we can talk about generalized weights of flags of Delsarte codes. The following observation makes them explicit. 

\begin{lemma}
Let $\mathsf{F} = (C_1, \dots , C_s)$ be a flag of Delsarte 
codes. 
Then the conullity function $\nusF$ of the associated $(q,m)$-demi-polymatroid $\PF =(E,\rhoF)$  is given by
$$
\nusF(X) = \sum_{i=1}^s(-1)^{i+1}\dim C_i(X)  \quad \text{for } X\in \Sigma(E).
$$
\end{lemma}

\begin{proof}
For $i=1, \dots , s$, let $\rho^{ }_i$ be as in \eqref{rhoFdefn} and let $\nu^*_i$ be the conullity function of the $(q,m)$-polymatroid $(E, \rho^{ }_i)$. Then in view of \eqref{nuC} in Remark \ref{nuCstar} we see that 
$$
\nusF(X) = \rhoF(E) - \rhoF(X^\perp) = \sum_{i=1}^s(-1)^{i+1}\nu^*_i(X) = \sum_{i=1}^s(-1)^{i+1}\dim C_i(X).
$$
for any $X\in \Sigma(E)$.
\end{proof}

We can now introduce the following generalization of 
Definition \ref{Def-drC}. 

\begin{definition}\label{Def-drF}
{\rm 
Let $\mathsf{F} = (C_1, \dots , C_s)$ be a flag of Delsarte 
codes in $\M$, and let $K = \rhoF(E) = \sum_{i=1}^s(-1)^{i+1}\dim C_i$. Then for $r=1, \dots , K$, the $r$th \emph{generalized weight} of $\F$ is denoted by $d_r(\F)$ or by $d_{\M,r}(C_1,\cdots,C_s)$, and is defined by 
$$
d_r(\F)
 =  \min \big\{\dim X : X\in \Sigma(E) \text{ with }    \sum_{i=1}^s(-1)^{i+1}\dim C_i(X) \ge r 
\big\}.
$$
}
\end{definition}
 
We 
remark that for $s=2$, these generalized weights were already defined by Mart\'{\i}nez-Pe\~{n}as and Matsumoto \cite[Definition 10]{MM},
and are referred to as \emph{RGMW profiles}, where RGMW stands for Relative Generalized Matrix Weights.
In \cite{MM} one studies these and and related profiles, and the interplay between them, in a way that carries the ideas and results of Luo,  Mitrpant, Han Vinck, and Chen \cite{LMVC} for pairs of block codes over to the world of Delsarte rank metric codes, in a way similar to the one, in which the relative profiles in \cite{LMVC}  are generalizations of those ``absolute ones" in the work of  Forney \cite{forney94} for single block codes.

Now that we have associated a $(q,m)$-demi-polymatroid $\PF =(E,\rhoF)$  to a flag $\F$ of Delsarte codes, it seems natural to ask whether 
$\PF^*$ is also a $(q,m)$-demi-polymatroid associated to some flag of Delsarte codes. The answer is yes, and it involves, quite naturally, the dual flag defined as follows. 

By the \emph{dual flag} corresponding to a flag  $\mathsf{F} = (C_1, \dots , C_s)$  of Delsarte codes, we mean the flag $\F^{\perp} = (C_s^\perp, \dots , C_1^\perp)$ of Delsarte codes, where $C_i^\perp$ is the trace dual of $C_i$ for $i=1, \dots , s$. Note that $C_1^{\perp} \subseteq \cdots \subseteq C_s^{\perp}$ so that $\F^\perp$ is indeed a flag in the sense of Definition \ref{RhoF}. Note also that $(\F^{\perp})^{\perp}=\F,$

Here is an analogue of \cite[Theorem 10]{BJM} for 
Delsarte rank metric codes, 

\begin{proposition} \label{dualflag}
Let $\F = (C_1, \dots , C_s)$ be a flag of Delsarte codes and $\F^\perp$ the dual flag corresponding to $\F$. Also let $\nusF$ denote the conullity function of the $(q,m)$-demi-polymatroid $\PF =(E,\rhoF)$ associated to $\F$. Then 
$$
\rhoFperp = \begin{cases} \rhoF^* & \text{if } s \text{ is odd}, \\
\nusF & \text{if } s \text{ is even}. \end{cases}
$$
\end{proposition}

\begin{proof}
A proof can be given, following word for word the proof of the corresponding result, \cite[Theorem 10]{BJM},
for linear block codes.
\end{proof}

 Proposition \ref{dualflag} identifies the dual $(q,m)$-demi-polymatroid of $\PF$ as that associated to the dual flag,
when $\F$ is a flag of odd length $s$, including the case $s=1$ (a case which trivially follows from Theorem \ref{commutes}).
But what about the cases when $s$ is even, including $s=2$, which perhaps are the most interesting ones? To this we remark that our study does not require  the Delsarte codes in the flag to be distinct. Thus, 
whenever $s$ is even, we can 
formally ``add" a subspace $C_{s+1}=\{ 0 \}$,
(irrespective of whether or not $C_S=\{ 0 \}$) to obtain a longer flag $\mathsf{G}$ of odd length $s+1$. 
Then it is easily seen that $\rhoG=\rhoF$, and using the duality for flags of even length, we obtain 
$$
\rhoF^*=\rhoG^*= 
{\rho_{\rule{0pt}{2ex}\{0\}^\perp }} - {\rho_{\rule{0pt}{2ex}C_{s}^\perp }} + \cdots + (-1)^s {\rho_{\rule{0pt}{2ex}C_{1}^\perp }} 
= {\rho_{\rule{0pt}{2ex}\M }} - {\rho_{\rule{0pt}{2ex}C_{s}^\perp }}  + \cdots + {\rho_{\rule{0pt}{2ex}C_{1}^\perp }} .
$$ 
In particular, in the important case $s=2$ studied 
in \cite{MM}, 
$$
\rhoF^* =  {\rho_{\rule{0pt}{2ex}\M }} - {\rho_{\rule{0pt}{2ex}C_{2}^\perp }} + {\rho_{\rule{0pt}{2ex}C_{1}^\perp }} 
$$ 

Furthermore, if one wants to organize the set of flags, into disjoint subsets, self-dual both with respect to $(q,m)$-demi-polymatroid duality, and the duality $\F \rightarrow \F^{\perp}$ of flags, we may as a convention first assume that all the Delsarte codes in each flag are distinct and nonzero. Then we get two cases, namely, flags of odd length and flags of even length. Now we modify
our convention and add the zero code as the innermost code in all even length flags. After this is done, all flags have odd length $2t+1$,
and all the Delsarte codes in each flag $\F=(C_1, \dots  , C_{2t+1})$ are distinct, and the possibility that both $C_{2t+1}= \{0\}$, and $C_1=M$ is perimitted. 
We call such flags as \emph{normalized flags}. We can then deduce from Proposition \ref{dualflag} the following.

\begin{proposition} \label{matchingduality}
Each $(q,m)$-demi-polymatroid associated to 
a flag (of Delsarte codes) in $\M$ comes from a unique normalized flag in $\M$ of odd length. For each $t=1,2, \cdots, \lfloor \frac{mn}{2}\rfloor$, the class of flags of length $2t+1$, and also its associated class of $(q,m)$-demi-polymatroids, is self-dual, and further, 
$$
\PF^* = P(\F^\perp) \quad \text{for every normalized flag } \F.
$$
\end{proposition} 

\begin{remark} \label{longflags}
{\rm 
If $\F$ is a normalized flag, then the longest possible length of $\F$ is clearly $mn+1$.  
The longest flag $(C_1, \dots , C_{mn+1})$ of distinct Delsarte codes in $\M$ will necessarily have $C_{mn+1}=\{0\}$, Hence if  $mn$ is odd, then it is not normalized, but if we delete $C_{mn+1}=\{0\}$ then it does become normalized and has odd length.  
Thus, the length of the longest normalized flag in $\M$ is $2t+1$, where $t=\lfloor \frac{mn}{2}\rfloor$.
 %
}
\end{remark}

As an immediate consequence of Theorem \ref{DemiDuality} and Proposition \ref{matchingduality}, we obtain a duality for the generalized weights of (normalized) flags of Delsarte codes. In the special case $s=2$ studied in \cite{MM}, this duality can be stated as follows.

\begin{corollary}
Let $C_1, C_2$ be distinct Delsarte codes in $\Mat$ with $C_2\subset C_1$. Then the relative generalized  weights
$$
d_r =\min \{\dim X : X \in \Sigma(E) \text{ with }  \dim C_1(X) - \dim C_2(X) \ge r\}
$$ 
are related to the relative generalized  weights 
$$
d_r^{\perp} =\min \{\dim X : X \in \Sigma(E) \text{ with }  \dim \M(X) - \dim C_2^{\perp}(X)+\dim C_1^{\perp}(X) \ge r\}
$$
via the ``$m$-fold" Wei duality described in Corollary \ref{important}, with $K=\dim C_1-\dim C_2.$
\end{corollary}

We end this paper by giving some examples.

\begin{example} \label{bigexample}
{\rm 
Let $n=3$ so that $E=\mathbb{F}_q^3$, and let $m=5$. 
The the full matrix space $\M$ has dimension $3 \times 5=15$. Let $X \subset Y$ be two subspaces of $E$ of dimension $1$ and $2$, respectively. 
Consider $C_2=\M(X)$ of dimension $1\times 5=5$, and $C_1=\M(Y)$ of dimension $2\times 5=10.$
For the flag $\F= (C_1, C_2)$, 
we then obtain for any $Z\in \Sigma(E)$, 
\begin{eqnarray*}
\rhoF(Z) &=& \dim \M(Y) - \dim \M(Y)(Z^{\perp}) - \big( \dim \M(X)-\dim \M(X)(Z^{\perp}) \big)\\
&= & 5\big( \dim Y - \dim X \big) - \big( \dim \M( Y \cap Z^{\perp})-\dim \M( X \cap Z^{\perp}) \big) \\
&=& 5\big(1 - \dim Y \cap Z^{\perp} + \dim  X \cap Z^{\perp} \big).
\end{eqnarray*}
We obtain a positive value (and that value is $5$) if and only if $Y \cap Z^{\perp} =  X \cap Z^{\perp}.$
This happens if and only if 
$Z + Y^{\perp} =  Z + X^{\perp}$ (cf. \cite[Lemma 5]{ravagnani15}). In case $\dim Z =1$, the latter happens if and only if 
$Z \subset X^{\perp}$, but $Z \ne  Y^{\perp}$.
For such $Z$, we see that the nullity function satisfies $\nuF(Z)= 5\dim Z-\rhoF(Z)=5-5=0$, but   
$\nuF(Z)= 5\dim Z-\rhoF(Z)=5-0=5$ for the other one-dimensional $Z$. 
For two-dimensional $Z$, the value of $\rhoF(Z)$ is positive if and only if $Z$ is not a plane intersecting $X^{\perp}$ in $Y^{\perp}$.
In this case, $\nuF(Z)= 5\dim Z-\rhoF(Z)=10-5=5$, whereas $\nuF(Z)= 5\dim Z-\rhoF(Z)=10-0=10$
for the other two-dimensional $Z$. For the unique $3$-dimensional space in $\Sigma(E)$, we find 
$\rhoF(E)=\rho^{ }_1(E)-\rho^{ }_2(E)=10-5=5, $ and so the rank $K$ of the $(q,5)$-demi-polymatroid is $5$, and $\nuF(E)=15-5=10.$
Hence the generalized weights $d_r$ of the dual $(q,5)$-demi-polymatroid $P(\F)^*$ satisfy  
$$
d_1=d_2=d_3=d_4=d_5=1 \quad \text{and moreover,} \quad d_6=d_7=d_8=d_9=d_{10}=2.
$$

Furthermore if $A$ and $B$ are two different planes through the origin, both 
intersecting $X^{\perp}$ in $Y^{\perp}$,
then $\rhoF(A)=\rhoF(B)=0$, while $\rhoF(A \cup B)=\rhoF(Y^{\perp})=0$ and 
$\rhoF(A + B)=\rhoF(E)=5.$
Hence axiom (R3) is violated, and $P(\F) = (E, \rhoF)$ is not a $(q,m)$-polymatroid. 
But of course, axioms (R1) and (R2) hold, and this can easily be checked directly. 

The dual normalized flag $\F^{\perp}$ is $(\M(X^{\perp}), \, \M(Y^{\perp}))$, and this is seen from  
the following result, which is easily verified:
$$
\M(A)^{\perp}=\M(A^{\perp}) \quad \text{for all subspaces $A$ of $E$, }
$$
where the first $\perp$ refers to orthogonality in the sense of trace duals, and the second $\perp$ refers to orthogonality for the standard dot product in $E$.
We then obtain for an arbitrary subspace $Z$ of $E$:
\begin{eqnarray*}
\rhoF^{\perp}(Z) &=& \dim \M - \dim \M(Z^{\perp}) - \big(\dim \M(X^{\perp})-\dim \M(X^{\perp} )(Z^{\perp} ) \big)  \\
&   & \phantom{\dim \M - \dim \M(Z^{\perp}) } +  \big(\dim \M(Y^{\perp}) - \dim \M(Y^{\perp})(Z^{\perp}) \big)\\
&= & 5\dim Z - 5 + 5\dim (X^{\perp} \cap Z^{\perp}) - 5 \dim (Y^{\perp} \cap Z^{\perp}).
\end{eqnarray*}
If $Z=\{0\}$, then this is $0-5+10-5=0$, and 
if $Z=E$, then it is $15-5+0-0=10$, which is the rank $mn-K=5\times 3 -5$ of $P(\F)^*$.
Hence ${\nu}^*(E)=15-10=5$,  as expected.
For $\dim Z =1$ or $2$, we can proceed 
as in the dual case above. 
Here $\rhoF^{\perp}(Z)$ is $5\dim Z - 5$ if $Z + Y=Z+X$, and it is $5\dim Z$ if
$Z+Y$ strictly contains $Z+X$. For $\dim Z=1$, we then get $\rhoF^{\perp}(Z)=0$ if and only if $Z$ is a line in $Y$ 
different from $X$, 
and $\rhoF^{\perp}(Z)=5$ otherwise. 
In other words $\nusF(Z)$ is nonzero (and equal to $5$)
if and only if $Z$ is a line in $Y$  different from $X$. 
For a $2$-dimensional subspace we obtain $\rhoF^{\perp}(Z)=5\dim Z -5 +0=5$ if $Z=Y$ or $Z$ is transversal to $X$, and 
 so ${\nu}^*(Z)=10-5=5$ then.
Moreover, $\rho_F^{\perp}(Z)=5\dim Z -5+5=10$ if $Z$ contains $X$ and is different from $Y$; hence ${\nusF}(Z)=10-10=0$ in this case. 
 One then easily  checks that (R1) and (R2) hold for $\rhoF ^*=\rhoF^{\perp}$, and therefore the condition (R4) holds for $P(\F)=(E,\rhoF)$, which indeed is a $(q,5)$-demi-polymatroid. From the determination of the values ${\nusF}(Z)$ above we see that for the generalized weights of the $(q,5)$-demi-polymatroid $P= P(\F)$ are given by 
$$
d_1=d_2=d_3=d_4=d_5=1.
$$ 
For an arbitrary $p$, modulo $5$, say $p=2$, let us check the values of 
$d_{p+im}(P^*)$ and $n+1-d_{p+K+jm}(P)=4-d_{p+K+jm}(P)$, where $P= \PF$. 
Since $K=5$ is divisible by $m=5$, we just check the $d_r(P)$ as well as $d_r(P^*)$ for $r\equiv 2 \pmod 5$. 
We see that 
$\{d_{p+im}(P^*)\}=\{d_2(P*), d_7(P^*)\}=\{1,2\}$, whereas 
$\{4-d_{p+K+jm}(P)\}=\{4-d_{2+jm}(P)\}=\{4-1\}=\{3\}.$
So these sets are disjoint and ``fill up" $\{1,2,3\}.$ The analogous statements of course hoøld also for $p$ congruent to $ 0,1,3,4$ modulo $5$.}
\end{example}

\begin{example} \label{previous}
{\rm Given $s$ vector subspaces $V_1,\dots,V_s$  of $E=\mathbb{F}_q^n$ and positive integers $m_1, \dots , m_s$ with $m_1+\cdots+m_s=m$, set $P=(E,\rho)$, where 
$$
\rho(J) =m_1 \rho_1(J)+\cdots+m_s\rho_s(J), \quad \text{where} \quad \rho^{ }_i(J)=\dim_{\mathbb{F}_q} (V_i  \cap J)
$$
for $i=1, \dots , s$ and $J\in \Sigma(E)$.
Then $P=(E,\rho)$ is,  in general, not a $(q,m)$-polymatroid.
As an example, take $n=2$,  $m=s=1$, and $\rho=\rho_1$.
Let $V_1$ be the diagonal $\{(a,a) \in E :  a \in \mathbb{F}_q\}.$ Also, let $J_1$ be the $x$-axis and $J_2$ the $y$-axis in 
$\mathbb{F}_q^2.$ Then $\rho(J_1)=\rho(J_2)=\dim (\{0\})=0$, but $\rho(J_1+J_2)= \dim (V_1  \cap E)=1$. Consequently,   $\rho(J_1 \cap J_2)+ \rho(J_1+J_2) >\rho(J_1)+\rho(J_2).$

The pair  $P=(E,\rho)$ is, however, a $(q,m)$-demi-polymatroid. Moreover, $P_i = (E, \rho^{ }_i) $ is a $(q,1)$-demi-polymatroid for each $i=1, \dots$. To see this, let us fix some $i\in \{1, \dots , s\}$. It is clear that $P_i$ 
satisfies axioms (R1) and (R2) with $m=1$. Further, for any $J\in \Sigma(E)$, we can write
$\rho^*_i(J)=\rho^{ }_i(J^{\perp})+\dim  J -\rho^{ }_i(E)$ as
$$
\dim (V_i\cap J^\perp) + \dim J - \dim V_i = \dim V_i + \dim J^\perp - \dim (V_i+J^\perp) + \dim J - \dim V_i .
$$
Since $\dim J^\perp + \dim J = \dim E$, it follows that
$$
\rho^*_i(J)= \dim E - \dim (V_i+J^\perp) =  \dim \big( (V_i+J^\perp)^\perp \big) = \dim (V_i^{\perp} \cap J).
$$
This shows that $\rho^*_i$ also satisfies axioms (R1) and (R2) with $m=1$. Thus $P_i$ is a $(q,1)$-demi-polymatroid.  As a consequence, we see that $\rho=\sum_{i=1}^s m_i\rho^{ }_i$ satisfies (R1) and (R2), and hence so does 
$\rho^*=\sum_{i=1}^s m_i\rho^{*}_i$. Thus $P=(E,\rho)$ is indeed a $(q,m)$-demi-polymatroid. The rank of $P$ is $K=\rho(E)=\sum_{i=1}^s m_i\dim V_i$. We observe that the ``independent" sets, i.e., those with ${\nu}(J)=m\dim J- \rho(J)=0$, are the ones that are contained in $V_1\cap \cdots \cap V_s$. 

For a $1$-dimensional $J\in \Sigma(E)$, we see that $\rho^*(J)$ is $\sum_{i=1}^s b_i$, where $b_i=m_i$ whenever  $V_i^{\perp}$  contains $J$, and $b_i=0$ otherwise. 
So, if there exists at least one $1$-dimensional $J$ for which at least one $V_i^{\perp}$ does not contain $J$, then we see that ${\nu}^*(J) \ge 1$, and the first generalized weight of $P$ is  $d_1=1$.  It follows that  $d_1(P)=1$, unless $V_i^{\perp}=E$ for all $i$, or equivalently, 
$\rho=0$.

We remark that the functions $\rho$ and $\rho^*$ are, in fact, the conullity and nullity functions of the $(q,m)$-polymatroid of Example \ref{sumsofmatroids} in disguise. Also, as noted in 
Proposition \ref{BasicDemi},  the  conullity and nullity functions of $(q,m)$-polymatroids give rise to $(q,m)$-demi-polymatroids, but not necessarily 
$(q,m)$-polymatroids. 
}
\end{example}

\section*{Acknowledgement}
We are grateful to Fr\'{e}d\'{e}rique Oggier for introducing us to rank metric codes and providing initial motivation for this work. We thank the DST in India and RCN in Norway for supporting our collaboration through an Indo-Norwegian project ``Mathematical Aspects of Information Transmission: Effective Error Correcting Codes". The last named author  is grateful to the Department of Mathematics, IIT Bombay for its warm 
hospitality, and for providing optimal conditions for research, during a visit in September--December 2018 when most of this work was done.

We also thank Elisa Gorla for her useful comments on the first version of this preprint and for pointing out the survey article  
\cite{G}. 
Those comments led us to make an appropriate revision in Theorem \ref{ardr} and Remark \ref{imprem}, and to introduce Proposition~\ref{demisquare}. 

\end{document}